\documentclass[10pt]{article}
\usepackage{amssymb}
\usepackage{amsmath}
\usepackage{amsfonts}
\usepackage{amsthm}

\usepackage[english]{babel}
\usepackage[affil-it]{authblk}

\newtheorem{theorem}{Theorem}
\newtheorem{lemma}{Lemma}
\newtheorem{proposition}{Proposition}
\newtheorem{corollary}{Corollary}

\title{On existence of perfect bitrades in Hamming graphs}

\author{I. Yu. Mogilnykh, F. I. Solov'eva %
\thanks{This work was funded by the Russian Science
Foundation under grant 18-11-00136.\\ 
\noindent E-mail address: \texttt{ivmog@math.nsc.ru, sol@math.nsc.ru}}}
\affil{Sobolev Institute of Mathematics\\
 Novosibirsk State University}
\begin{document}

\maketitle




\begin{abstract}

A pair $(T_0,T_1)$ of  disjoint sets of vertices of a graph $G$ is called a {\it
perfect bitrade} in $G$ if any ball of radius 1 in $G$ contains exactly one vertex in $T_0$ and $T_1$ or none simultaneously.
The {\it volume} of a perfect bitrade $(T_0,T_1)$ is the size of $T_0$. In particular, if $C_0$ and $C_1$ are distinct perfect codes with minimum distance $3$ in $G$ then
$(C_0\setminus C_1,C_1\setminus C_0)$ is a perfect bitrade.  For any $q\geq 3$, $r\geq 1$  we construct perfect
bitrades in the Hamming graph $H(qr+1,q)$ of volume $(q!)^r$ and show that for $r=1$ their volume is minimum.
\end{abstract}

\noindent{\bf Keywords:}
perfect code, one-error-correcting code, trade, bitrade, spherical bitrade, perfect bitrade, MDS code, alternating group




\section{Introduction}

Bitrades are used 
for constructing large classes of codes and designs
and investigating nontrivial
structural properties of these combinatorial objects. It should be
noted that in general bitrades are defined independently on
including them into  codes or designs. Bitrades could exist even
regardless of existence of the parent objects with the corresponding parameters. These facts provide a serious additional
motivation for constructing bitrades and studying their
properties. Classical problems in this area are the existence of
bitrades and bounds on their volumes.

In the paper we consider the problem of constructing perfect
bitrades, which concerns to the classical problem of existence of
perfect codes in the non prime power case. In 1973 Zinoviev and
Leontiev \cite{ZL73} and independently Tiet\"{a}v\"{a}inen
\cite{Tiet} proved that if $q$ is a power of a prime number then
there are only perfect codes with the parameters of $q$-ary
Hamming codes, binary and ternary Golay codes. The number of
perfect one-error-correcting codes is double exponential
\cite{AK}, however full classification and enumeration are still
open problems. By attempts of several authors it was proved that
for minimum distances more than 5 there are no perfect codes over
a non prime alphabets. In \cite{BZLF} Bassalygo at al. established
the nonexistence of perfect codes for $q=2^i3^j$, $i,j\geq 1$ with
minimum distance at least 5. In 1964 Golomb and Posner proved the
nonexistence of perfect codes of length 7 with minimum distance 3 over the alphabet of six elements. More information could be found in
\cite{Zin_diss} and in the survey \cite{H2008} with the lists of
references there.  A special case of bitrades arize from
components ($i$-components, $\alpha$-components) of codes. We
refer to \cite{S2008Switch} for a survey on the switchings of
$i$-components in perfect codes, see also \cite{Ost,PL99}. We note
that a concept of bitrades  was developed for MDS codes, see a
work of Potapov \cite{Pot}.

{\it A Steiner $(w-1,w,n)$-bitrade} $(T_0,T_1)$ is defined as a
pair of disjoint collections $T_0$ and $T_1$ of $w$-subsets of a
$n$-element point set, such that any $(w-1)$-subset of the point
set is a subset of exactly one set in $T_0$ and $T_1$ or none. A
lower bound on the volume (i.e. the size of $|T_0|$) of such
bitrade is obtained by Hwang in \cite{Hwang} along with a
characterization for the bitrades of minimum volume. For a survey
on Steiner and other related bitrades we refer to \cite{HK}.
Krotov et al. \cite{KMP} suggested a generalization of this concept
for $q$-ary Steiner bitrades and established an attainable lower bound for their volumes.





For perfect bitrades that are embedded into
perfect codes in $H(n,2)$ the lower bound $2^{(n-1)/2}$ for their
volumes is known, see Etzion and Vardy \cite{EV98} or Solov'eva
\cite{S}. The argument of \cite{EV98} holds for the perfect
bitrades regardless of being embedded into a perfect code and for
any odd $n$ not necessarily being a power of two but one. The
bound is attained on so-called minimum $i$-components, that were
used for constructing first nonlinear perfect binary codes by
Vasil'ev. A classification of perfect bitrades in the binary case
was obtained by Krotov for $n=9$  in  \cite{Kro}.

Given a perfect bitrade $(T_0,T_1)$ in the Hamming graph $H(n,q)$,
it is not hard to see that the vector $\chi_{T_0}-\chi_{T_1}$ is
an eigenvector of the adjacency matrix of $H(n,q)$ with the eigenvalue
$-1$. Here $\chi_{T_0}$ and $\chi_{T_1}$ denote the characteristic
vectors of $T_0$ and $T_1$ respectively in the vertex set of
$H(n,q)$. This fact relates the problem of determining the minimum
volume of perfect bitrades to the problem of finding eigenvectors
of $H(n,q)$ with minimum size of support. For eigenvalue $-1$ this
problem was firstly considered in \cite{KroVor} and solved for the
eigenvalue $(n-1)q-n$ in \cite{V} by Valyuzhenich. The approach of
work \cite{V} was further extended in \cite{VV} for arbitrary
eigenvalue of the Hamming graph with complete solution for all
eigenvalues in case when $q$ is at least $4$.
 In particular, the result \cite{VV} implies a lower bound $2^{n-\frac{n-1}{q}}(q-1)^{\frac{n-1}{q}}$  for the volume of the perfect bitrades in $H(n,q)$, $q\geq 4$.
A construction for perfect bitrades from \cite{KroVor}
 gives the upper bound $2^{\frac{n-1}{q}+1}q^{\frac{(n-1)(q-2)}{q}}$ on the minimum volume of a bitrade in $H(n,q)$ for $q$ being powers of primes.

The current paper is organized as follows. Basic definitions and a
revision of some previous results are given in Section 2. In
particular, we introduce a concept of a spherical bitrade which is
crucial for obtaining the main results of the paper. Basic theory
regarding spherical and perfect bitrades is presented in Section
3. We reveal interrelations between perfect and spherical bitrades
and eigenfunctions of Hamming graphs and obtain a natural
recursive construction for spherical bitrades. In particular, the
results of Section 3 allow to reduce the problem of constucting
perfect bitrades in $H(qr+1,q)$ to constructing spherical bitrades
in $H(q,q)$. The latter is solved in Section 4 as we split the
opposite of the repetition code by the parity corresponding
permutation in order to obtain bitrades. Thus we construct
perfect bitrades in the Hamming graph $H(qr+1,q)$ of the volume
$(q!)^r$ for any $q\geq 3$, $r\geq 1$.
In Section 5 using a combinatorial argument we show that the
actual value of the minimum volume of perfect bitrade in
$H(q+1,q)$ is $q!$.

\section{Definitions and preliminaries}

The vertex set of the {\it Hamming graph} $H(n,q)$
consists of the tuples of length $n$ over the alphabet set
$\{0,\ldots,q-1\}$ which we denote by ${\cal A}$ and tuples $x$
and $y$ are adjacent if they differ in exactly one coordinate
position. A {\it code} in a graph $G$ is a subset of its vertices.
The {\it size} of $C$ is $|C|$ and its
 {\it minimum distance} $d_C$ is $d_C=min_{x,y\in C,x\neq y}d(x,y)$, where $d(x,y)$ is the
length of a shortest path connecting $x$ and $y$. A code $C$ is
{\it one-error-correcting perfect} (in throughout what follows
perfect) if the balls of radius 1 centered at the vertices of $C$
part the vertex set of $G$. If $C$ and $D$ are  codes in $G$, $d(C,D)$ denotes $min\{d(x,y):x \in C,y \in D\}$.

When $C$ is a code in the Hamming graph $H(n,q)$ we use a
traditional expression {\it $q$-ary code of length $n$}. If $q$ is
a power of  a prime, the alphabet set ${\cal A}$ is associated with the Galois
field $F_q$ of order $q$. In this case  we consider the linear space $F_q^n$ on
the set vertices of $H(n,q)$ naturally inherited from $F_q$. With
this regard, a code is called {\it linear}, if it is a linear
subspace of $F_q^n$. The well-known {\it Singleton bound} states
that the size of a q-ary code $C$ of length $n$ with minimum
distance $d_C$ is not greater than $q^{n-d_C+1}$. If the size of
the code $C$ attains the Singleton bound, $C$ is called a
{\it MDS code}.

Let $x$ be a vertex of a graph $G$. Denote by ${\cal O}(x)$ the
sphere of radius one centered at $x$, i.e. the set of the
neighbors of $x$ in $G$. Let $S_0$ and $S_1$ be two disjoint codes
in $G$. The ordered pair $(S_0,S_1)$ is called a {\it spherical
bitrade} in $G$ if for any vertex $x$ of $G$ we have
$$|{\cal O}(x)\cap S_0|=|{\cal O}(x)\cap S_1|\in \{0,1\}. $$
Denote by $\overline{{\cal O}}(x)$ the ball of radius one centered
at a vertex $x$ in $G$, i.e.  $\overline{{\cal O}}(x)={\cal
O}(x)\cup \{x\}$. The ordered pair $(S_0,S_1)$ is called a {\it
perfect bitrade} in $G$ if for any vertex $x$ of $G$ we have
$$|\overline{{\cal O}}(x)\cap S_0|=|\overline{{\cal O}}(x)\cap S_1|\in \{0,1\}. $$
The {\it volume} of a spherical or perfect bitrade $(S_0, S_1)$ is
$|S_0|$ (or $|S_1|$).

A real-valued function $f$ defined on the vertex set of a graph $G$ is called a $\lambda$-{\it eigenfunction}, if it is not the all-zero function and
$$\lambda f(x)=\sum_{y\in {\cal O}(x)} f(y).$$
In other words, the vector of the values  of $f$ is an eigenvector
of the adjacency matrix of the graph $G$ with the eigenvalue
$\lambda$. It is well-known that the eigenvalues of the adjacency
matrix of the Hamming graph $H(n,q)$ is the  set $\{n(q-1)-qi:i
\in \{0,\ldots,n\}\}$.

For a pair of tuples $x$ and $y$  over the same alphabet define
their concatenation by $x|y$. If $C$ and $D$ are two codes then
denote by $C\times D$  the following code $\{x|y:x \in C, y \in
D\}$. Let $f$ and $g$ be two real-valued functions on the vertices
of $H(n,q)$ and $H(n',q)$ respectively.
 The {\it tensor product} of $f$ and $g$ is the function $f\cdot g$ defined on the set of vertices of $H(n+n',q)$ as follows:
$(f\cdot g)(x|y)=f(x)g(y)$.

\begin{lemma}\label{lemVV} \cite{VV}[Corollary 1] Let $f$  be a $\lambda$-eigenfunction of $H(n,q)$ and  $g$ be a $\mu$-eigenfunction of $H(n',q)$.
  Then the function $f\cdot g$ is a $(\lambda+\mu)$-eigenfunction of $H(n+n',q)$.
\end{lemma}

Let $i_1,\ldots,i_k$ be pairwise distinct coordinate positions,
$\{i_1,\ldots,i_k\} \subseteq \{1,\ldots,n\}$ and  $a_1,\ldots,
a_k$ be symbols of the alphabet set ${\cal A}$. The set of all tuples $x$ of length $n$ over ${\cal A}$ such that
$x_{i_l}=a_l \mbox{ for any }l\in \{1,\ldots,k\}$
 is called a {\it face} in $H(n,q)$ and is denoted by $\Gamma_{i_1\ldots i_k}^{a_1\ldots a_k}$. Any position from
$\{1,\ldots,n\}\setminus \{i_1,\ldots,i_k\}$ in the face
$\Gamma_{i_1\ldots i_k}^{a_1\ldots a_k}$ is called {\it free}.
 We finish the preliminary part of the paper with a well-known
result of Delsarte.

\begin{theorem}\cite{Del}
Let $f$ be a $(n(q-1)-mq)$-eigenfunction of $H(n,q)$ for $m\in
\{0,\ldots,n\}$. Then for any $k\leq m-1$, any pairwise distinct
elements $i_1,\ldots,i_k$ of $\{1,\ldots,n\}$ and symbols $a_1,\ldots,a_k$ of ${\cal A}$ the sum of the values of $f$ on the face
$\Gamma_{i_1\ldots i_k}^{a_1\ldots a_k}$ is zero.
\end{theorem}

\begin{corollary}\label{coroDel}
Let $f$ be a $(n(q-1)-mq)$-eigenfunction of $H(n,q)$ for $j\in \{0,\ldots,n\}$ . Then
 for any pairwise distinct elements $i_1,\ldots,i_{m-1}$  of $\{1,\ldots,n\}$ and symbols $a_1,\ldots,a_k$ of ${\cal A}$ the function $f$ has at least two nonzero
 values on the face $\Gamma_{i_1\ldots i_{m-1}}^{a_1\ldots a_{m-1}}$ or all values of $f$ on $\Gamma_{i_1 \ldots i_{m-1}}^{a_1 \ldots a_{m-1}}$  are zeroes.

\end{corollary}

\section{Spherical and perfect bitrades}
\subsection{Bitrades and eigenfunctions}
We now give a characterization for  perfect bitrades in graphs and
spherical bitrades in Hamming graphs in terms of eigenfunctions of
these graphs.

\begin{proposition}\label{prop_perf}
Let $T_0$ and $T_1$ be disjoint codes in  a graph $G$ of diameter at least
$3$. The ordered pair $(T_0,T_1)$ is a perfect bitrade in
$G$ if and only if $d_{T_0}=d_{T_1}=3$ and $\chi_{T_0}-\chi_{T_1}$
is a $(-1)$-eigenfunction of $G$.
\end{proposition}

\begin{proof}
Let for a pair of disjoint codes $T_0$ and $T_1$  the equality
$|\overline{{\cal O}}(x)\cap T_0|=|\overline{{\cal O}}(x)\cap
T_1|$ be fulfilled for any vertex $x$ of $G$. We see that this
property holds if and only if
$-(\chi_{T_0}-\chi_{T_1})(x)=\sum\limits_{y\in {\cal
O}(x)}\chi_{T_0}(y)-\chi_{T_1}(y)$ for any  vertex $x$ of $G$,
i.e. $\chi_{T_0}-\chi_{T_1}$ is a (-1)-eigenfunction of $G$.

Obviously, the balls of radius one centered at the vertices of
a code with minimum distance $d$ are disjoint if and only
if $d$ is at least 3. Let $(T_0,T_1)$
be a perfect bitrade, $x$ and $y$ be neighbors from $T_0$ and
$T_1$. Let $z$ be a neighbor of $x$ but not a neighbor of $y$. By
the definition of a perfect bitrade, $z$ has a neighbor in $T_1$
at distance 3 from $y$.\end{proof}

A graph $G$ of diameter $d$ is called {\it distance-regular}, if there are constants $p_{ij}^k, i,j,k\in\{0,\ldots,d\}$ such that for any pair of vertices $x, y$, $d(x,y)=k$ $p_{ij}^k=|\{z:d(x,z)=i, d(z,y)=j\}|$.
In throughout what follows $\oplus$ denotes the addition via modulo 2.

\begin{lemma}\label{Lemma_sph1} Let $G$ be a distance-regular graph with $p_{11}^1, p_{21}^2\neq 0$. We have the following:

1. Let $C$ be a code in $G$. Then $|{\cal O}(x)\cap
C|\leq 1$ if and only if $d_C\geq 3$.

2. Let $S_0$ and $S_1$ be disjoint codes in $G$.
Then $\chi_{S_0}-\chi_{S_1}$  is a $0$-eigenfunction of $G$
if and only if $|{\cal O}(x)\cap S_0|=|{\cal O}(x)\cap S_1|$ for
any vertex $x$ of $G$.

3.  Let 
$(S_0,S_1)$ be a spherical bitrade in $G$ or
$d_{S_0},d_{S_1}\geq 3$, $\chi_{S_0}-\chi_{S_1}$ be a
$0$-eigenfunction of $G$. Then $d_{S_0}=d_{S_1}=3$.

\end{lemma}

\begin{proof}
1. The sufficiency is clear. If there are vertices from $C$ at distance 1 or 2 in $G$, then they obviously have a common neighbor $x$, so $|{\cal O}(x)\cap C|\geq 2$.

2. Follows from the definition of a $0$-eigenfunction of $G$.

3. Since $p_{21}^2$ is nonzero there is a path of length 4: $x, x^1, x^2, x^3$ in $G$
such that $x\in S_i$, $x^2\in S_{i\oplus 1}$, $d(x^3,x)=2$, $i\in
\{0,1\}$. The vertex $x^3$ is a neighbor of $x^2$ from $S_{i\oplus
1}$. By the definition of a spherical bitrade and the second
statement of the current lemma $x^3$ must have a neighbor from
$S_i$. The latter is at distance 3 from $x$ because $d_{S_i}\geq
3$ (see the first statement of
the current lemma), so $d_{S_i}=3$.
\end{proof}


\begin{proposition}\label{prop_sph1}
 Let $G$ be a distance-regular graph with $p_{11}^1, p_{21}^2\neq 0$. The following statements are equivalent for codes $S_0$ and $S_1$ in $G$:

i.  The pair $(S_0,S_1)$ is a spherical bitrade in $G$.

ii.   The minimum distances of $S_0$ and $S_1$ are 3 and
$\chi_{S_0}-\chi_{S_1}$ is a $0$-eigenfunction of $G$.

iii. The minimum distances of $S_0$ and $S_1$ are 3,
$d(S_0,S_1)=2$  and for any $x\in S_i$ there are exactly
$p^0_{11}/p_{11}^2$ vertices in $S_{i\oplus 1}$ at distance 2 from $x$.
\end{proposition}

\begin{proof}
(i)$\sim$(ii) Follows from the first and the second statements of Lemma \ref{Lemma_sph1}.



(i)$\sim$(iii) Let $(S_0,S_1)$ be a spherical bitrade. Then by the
third statement  of Lemma \ref{Lemma_sph1}  we have
$d_{S_0}=d_{S_1}=3$.
 Suppose that  $x\in S_0$ is a neighbor of $y\in S_1$. Then by the definition of a spherical bitrade, $x$ has a neighbor  from $S_0$, which contradicts the first statement of Lemma \ref{Lemma_sph1}.
 Therefore, $d(S_0,S_1)$ is $2$.

Now let $S_0$ and $S_1$ be two codes with minimum distances three such that  $d(S_0,S_1)=2$. A vertex $x$ from $S_0\cup S_1$ has no neighbors in $S_0$ and $S_1$ because  $d_{S_0}=d_{S_1}= 3$, $d(S_0,S_1)=2$.
\begin{equation}\label{simple}
\mbox{ If x is in }S_0\cup S_1 \mbox{ then }|{\cal O}(x)\cap S_0|=|{\cal O}(x)\cap S_1|=0.
\end{equation}

Given a vertex $x$ from $S_i$ let us consider the set $\{y\in
S_{i\oplus 1}:d(x,y)=2\}$,  $i\in \{0,1\}$. Each of the vertices
from this set has exactly $p_{11}^2$ common neighbors with $x$.
Moreover, distinct vertices from $\{y\in S_{i\oplus 1}:d(x,y)=2\}$
have disjoint sets of common neighbors with $x$ because
$d_{S_{i\oplus 1}}=3$. Now each of the neighbors of $x$ is a
neighbor of exactly one vertex from $S_{i\oplus 1}$ if and only if
there are exactly $|{\cal O}(x)|/p_{11}^2=p^0_{11}/p_{11}^2$
vertices from $S_{i\oplus 1}$ at distance 2 from $x$. Taking into
account the property (\ref{simple}), the proposition follows.
\end{proof}

\medskip

In nonbipartite case we have the following characterization for Hamming graphs.

\begin{corollary}\label{prop_sph}
The following assertions are equivalent for $q$  greater or equal
to $3$ and for two disjoint $q$-ary codes $S_0$ and $S_1$ of length
$n$:

i.  The pair $(S_0,S_1)$ is a spherical bitrade in $H(n,q)$.

ii.   The minimum distances of $S_0$ and $S_1$ are 3 and
$\chi_{S_0}-\chi_{S_1}$ is a $0$-eigenfunction of $H(n,q)$.

iii. The minimum distances of $S_0$ and $S_1$ are 3,
$d(S_0,S_1)=2$  and for any $x\in S_i$ there are exactly
$(q-1)n/2$ tuples of $S_{i\oplus 1}$ at distance 2 from $x$.
\end{corollary}

The eigenvalues of the Hamming graph are $\{n(q-1)-qi:i \in \{0,\ldots,n\}\}$. Then taking into account the eigenfunction representations given in Proposition \ref{prop_perf} and Corollary \ref{prop_sph} we see that a spherical (perfect respectively) bitrade exists in $H(n,q)$ then necessarily $n$ is $qr$ ($qr+1$ respectively), for some $r\geq 1$.


\subsection{Perfect bitrades from spherical bitrades}

\begin{proposition}\label{prop_sphperf}
 Let $(S_0,S_1)$ be a spherical bitrade in $H(qr,q)$. Then  $(S_0\times \{0\} \cup S_1\times \{1\},S_0\times \{1\} \cup S_1\times \{0\})$ is a perfect bitrade.
\end{proposition}
\begin{proof}
By Corollary \ref{prop_sph} the function $\chi_{S_0}-\chi_{S_1}$
is a $0$-eigenfunction of $H(qr,q)$, $d(S_0,S_1)=2$ and $d_{S_0}$
and $d_{S_1}$  are three. Consider the difference of characteristic functions $\chi_{0}$ and $\chi_1$ of vertices (symbols) 0 and 1 in the complete graph $H(1,q)$.
 It is clear that $\chi_0-\chi_1$ is a
$(-1)$-eigenfunction of $H(1,q)$. By Lemma \ref{lemVV} we conclude
that  $$(\chi_{S_0}-\chi_{S_1})\cdot (\chi_0-\chi_1)=\chi_{S_0\times\{0\}\cup
S_1\times \{1\}}-\chi_{S_0\times\{1\}\cup S_1\times \{0\}}$$ is a
$(-1)$-eigenfunction of $H(qr+1,q).$  The minimum distances
of $S_0$ and $S_1$ are three and the distances between vertices of
$S_0$ and $S_1$ are at least two. This implies that the minimum distances
of $S_0\times\{0\}\cup S_1\times \{1\}$ and $S_0\times\{1\}\cup
S_1\times \{0\}$ are three. By Proposition \ref{prop_perf} we
conclude that $(S_0\times \{0\} \cup S_1\times \{1\},S_0\times
\{1\} \cup S_1\times \{0\})$ is a perfect bitrade.
\end{proof}

\subsection{A recursive construction for spherical bitrades}

\begin{theorem}\label{TheoRecSph}
Let $(S_0,S_1)$ and $(S'_0,S'_1)$ be spherical bitrades in $H(qr,q)$ and $H(qr',q)$. Then $(S_0\times S'_0 \cup S_1\times S'_1,S_0\times S'_1 \cup S_1\times S'_0)$ is a spherical bitrade in $H(q(r+r'),q)$.

\end{theorem}
\begin{proof}
By Corollary \ref{prop_sph}  the functions
$\chi_{S_0}-\chi_{S_1}$ and $\chi_{S'_0}-\chi_{S'_1}$ are
$0$-eigenfunctions of $H(qr,q)$ and $H(qr',q)$ and the minimum
distances of $S_0$, $S_1$, $S'_0$, $S'_1$ are three. Consider the
tensor product
$(\chi_{S_0}-\chi_{S_1})\cdot(\chi_{S'_0}-\chi_{S'_1})$. By Lemma
\ref{lemVV} this function is a $0$-eigenfunction of $H(q(r+r'),q)$. We have the
following equalities:

$$(\chi_{S_0}-\chi_{S_1})\cdot(\chi_{S'_0}-\chi_{S'_1})=(\chi_{S_0}\cdot\chi_{S'_0})-(\chi_{S_1}\cdot\chi_{S'_0})-(\chi_{S_0}\cdot\chi_{S'_1})+(\chi_{S_1}\cdot\chi_{S'_1})=$$
$$=\chi_{S_0\times S'_0}+\chi_{S_1\times S'_1}-\chi_{S_0\times S'_1}-\chi_{S_1\times S'_0}=\chi_{S_0\times S'_0\cup S_1\times S'_1}-\chi_{S_0\times S'_1\cup S_1\times S'_0}.$$

We see that the function $\chi_{S_0\times S'_0\cup S_1\times
S'_1}-\chi_{S_0\times S'_1\cup S_1\times S'_0}$ is a
$0$-eigenfunction of $H(q(r+r'),q)$. Moreover it is easy to see
that the minimum distances of $S_0\times S'_0\cup S_1\times S'_1$
and $S_0\times S'_1\cup S_1\times S'_0$  are also three which
follows from the minimum distances of
 $S_0$, $S_1$, $S'_0$, $S'_1$. By Corollary \ref{prop_sph}
 the pair
 $(S_0\times S'_0 \cup S_1\times S'_1,S_0\times S'_1 \cup S_1\times S'_0)$   is a spherical bitrade.
\end{proof}

\section{Constructions of spherical and perfect bitrades}

It is a well-known fact that the action
of any automorphism of $H(n,q)$ can be represented as the action of a
permutation $\pi$ on the coordinate positions $\{1,\ldots,n\}$ followed by the
action of $n$ permutations $\sigma_{1},\ldots,\sigma_{n}$ of
the alphabet set ${\cal A}$:
$$\pi(x_{1},\ldots,x_{n})=(x_{\pi^{-1}(1)},\ldots,x_{\pi^{-1}(n)}),$$
$$(\sigma_{1},\ldots,\sigma_{n})(x_{1},\ldots,x_{n})=(\sigma_{1}(x_{1}),\ldots,\sigma_{n}(x_{n})).$$


Let $Sym_{\cal A}$ and $Alt_{\cal A}$ denote the symmetric and alternating
groups on the elements of the set ${\cal A}$. Define the codes $S_0$ and
$S_1$ as follows:

\begin{equation} \label{S0S1}
\begin{split}
S_0=\{(\pi(0),\pi(1),\ldots,\pi(q-1)): \pi   \in Alt_{\cal A}\} \\
S_1=\{(\pi(0),\pi(1),\ldots,\pi(q-1)): \pi
\in Sym_{\cal A}\setminus Alt_{\cal A}\}
\end{split}
\end{equation}

Note that
$$S_1=\{(\pi(1),\pi(0),\pi(2),\ldots,\pi(q-1)): \pi   \in Alt_{\cal A}\}.$$
So the group $\{(\pi,\pi,\ldots,\pi): \pi   \in Alt_{\cal A}\}$ acts
regularly on the tuples of the codes $S_0$ and $S_1$ and the group
$\{(\pi,\pi,\ldots,\pi): \pi \in Sym_{\cal A}\}$ acts regularly on the
tuples of  $S_0\cup S_1$. Moreover, the code $S_1$ is
obtained from $S_0$ by an authomorphism of $H(q,q)$, i.e. a
transposition of coordinate positions.

\begin{theorem} \label{SphericalBitradesAlt}
Let $S_0$ and $S_1$ be the codes defined by (\ref{S0S1}). Then
$(S_0,S_1)$ is a spherical bitrade in ${\mathcal H}(q,q)$ of volume $q!/2$.
\end{theorem}

\begin{proof}
We prove that $(S_0,S_1)$ fulfills the statement (iii) of
Corollary \ref{prop_sph},  i.e. $d_{S_i}=3$, where $i\in
\{0,1\}$; $d(S_0,S_1)=2$ and for any $x\in S_i$ there exist
$(^q_2)$ tuples  in $ S_{i\oplus1}$ at distance 2 from $x$ for any
$i\in \{0,1\}$.

Let $x$ be the tuple $(0,1,\ldots,q-1)$, $x \in S_0$. If $\pi$ is
in $Sym_{\cal A}$ then the distance between the tuples $(0,1,\ldots,q-1)$
and $(\pi(0),\pi(1),\ldots,\pi(q-1))$ is always at least 2 and
equals 2 if and only if $\pi$ is a transposition. So by
(\ref{S0S1}) the only tuples from $S_0 \cup S_1$ at distance 2
from $x$ are $(^q_2)$ tuples of $S_1$. Since $S_0$ and $S_1$ are
orbits of the same group and the code $S_1$ is obtained from $S_0$
by an authomorphism of $H(q,q)$
we have the
desired properties.
\end{proof}

\begin{theorem}
For any integer $r\geq 1$ and $q\geq 3$ there is a spherical bitrade in $H(qr,q)$ of volume $(q!)^r/2$ and a perfect bitrade in $H(qr+1,q)$ of volume $(q!)^r$.

\end{theorem}
\begin{proof}
In order to obtain a spherical bitrade in $H(qr,q)$ we apply
$(r-1)$ times the construction of Theorem \ref{TheoRecSph} with
the  initial bitrade being the spherical bitrade in $H(q,q)$ from
Theorem \ref{SphericalBitradesAlt}. This spherical bitrade in
$H(qr,q)$ implies the existence of a perfect bitrade according to
Proposition \ref{prop_sphperf}.
\end{proof}

\smallskip

In the linear case we also have bitrades from MDS codes of significantly larger volumes than those of constructed in Theorem \ref{SphericalBitradesAlt}.

\begin{theorem} \label{SphericalTradesMDS}
Let $q$ be $p^r$, p be a prime number. Let $M$ be a $q$-ary linear MDS
code with minimum distance 2, $C_0$ and $C_1$ be $q$-ary linear
MDS  codes of length $n$ with minimum distances 3,
 $C_0,C_1\subset M$, $C_0\neq C_1$. Then $(C_0\setminus C_1,
C_1\setminus C_0)$ is a spherical bitrade in ${\mathcal H}(q,q)$
of volume $q^{q-2}-q^{q-3}$.
\end{theorem}

\begin{proof}
We show that $|{\cal O}(x)\cap C_0|=|{\cal O}(x)\cap C_1|$ is zero if $x \in M$ and one otherwise.
 Since $M$ has minimum distance 2, $C_0$ and $C_1$ are subcodes of $M$ we see that any tuple of $M$ has no neighbors in $C_0$ and $C_1$.
It remains to show that if $x$ is not in $M$ then $|{\cal O}(x)\cap C_0|=|{\cal O}(x)\cap C_1|=1$.
Since $d_{C_0}=d_{C_1}=3$ we  have that $|{\cal O}(x)\cap C_0|,|{\cal O}(x)\cap C_1|\leq 1$.
 Then the  number of the neighbors of $C_i$ could be counted as
$$|C_i|\cdot q(q-1)=q^{q-2}\cdot q(q-1)=q^q-|M|.$$
The neighbors of $C_i$ cannot be in $M$, so the above implies that each of the tuples outside of $M$ is a neighbor of only one tuple in $C_0$ and in $C_1$. Since linear codes $C_0$ and $C_1$ are of dimension $q-2$, they meet in a subspace of dimension $q-3$ and the expression for the volume of the bitrade $(C_0\setminus C_1,
C_1\setminus C_0)$ follows.
\end{proof}

\begin{theorem}
For any integer $r\geq 1$ and $q\geq 3$, $q=p^r$, where $p$ is a prime there is a spherical bitrade in $H(qr,q)$ of volume $2^{r-1}(q^{q-2}-q^{q-3})^r$
and a perfect bitrade in $H(qr+1,q)$ of volume $2^{r}(q^{q-2}-q^{q-3})^r$.

\end{theorem}
\begin{proof}
A spherical bitrade in $H(qr,q)$ of volume $2^{r-1}(q^{q-2}-q^{q-3})^r$ is obtained by Theorem \ref{TheoRecSph} from the
spherical bitrade in $H(q,q)$ described in Theorem \ref{SphericalTradesMDS}. We then apply the construction from Proposition
\ref{prop_sphperf} to the spherical bitrade in $H(qr,q)$ in order to obtain a perfect bitrade in $H(qr+1,q)$.
\end{proof}

\medskip

\noindent
{\bf Remark.} There are other spherical bitrades that
have rather less symmetric structure than the described above. In
$H(5,5)$ using a computer we have found spherical bitrades of the
following volumes: 60, 95, 100, 125. A bitrade of volume 60 can be
obtained by Theorem \ref{SphericalBitradesAlt}, a bitrade of
volume 100 exists by Theorem \ref{SphericalTradesMDS}. A spherical
bitrade of volume 125 could be obtained by taking $C_1$ to be a
coset of a linear MDS code $C_0$ in Theorem
\ref{SphericalTradesMDS}, the proof for this fact is the same.
This bitrade coincides with the bitrade described in the work
\cite{KroVor}.

\section{Lower bound for the volumes of perfect bitrades in  ${\mathcal H}(q+1,q)$}

\begin{theorem}\label{thm_perflow}
The volume of a perfect bitrade in ${\mathcal H}(q+1,q)$ is not less than $q!$.
\end{theorem}
\begin{proof}
Let $(T_0,T_1)$ be a perfect bitrade in ${\mathcal H}(q+1,q)$, $\Gamma$ be a face in ${\mathcal H}(q+1,q)$  with $(k+1)$, $q\geq k\geq 1$ free positions.
We show that  $|\Gamma\bigcap
T_0|+|\Gamma\bigcap T_1|\geq 2(k!)$ or $|\Gamma\bigcap
T_0|+|\Gamma\bigcap T_1|=0$. The proof is by induction on the number of free positions of $\Gamma$.

The base case is when $k$ is $1$. Let $\Gamma$ be a face in
${\mathcal H}(q+1,q)$  with exactly $2$ free positions. By Proposition
\ref{prop_perf} the function $\chi_{T_0}-\chi_{T_1}$ is a
$(-1)$-eigenfunction of $H(q+1,q)$. Taking into account Corollary
\ref{coroDel} the function $\chi_{T_0}-\chi_{T_1}$ has 0 or at
least 2 nonzero values on $\Gamma$, i.e. $|\Gamma\bigcap
T_0|+|\Gamma\bigcap T_1|$ is 0 or at least 2.

By induction hypothesis we have that any face with exactly $k$
free positions has at least $(k-1)!\cdot 2$ tuples of $T_0\cup
T_1$ or none. Without restriction of generality, suppose that the
face $\Gamma=\Gamma_{1\ldots q-k}^{0\ldots 0}$ in ${\mathcal
H}(q+1,q)$ contains the all-zero tuple $\bf{0}$ from $T_0$. We
show that there are at least $2(k!)$ tuples from $T_0\bigcup T_1$
in $\Gamma_{1\ldots q-k}^{0\ldots 0}$.

By the definition of a perfect bitrade, the ball centered at ${\bf
0}\in T_0$ must contain a unique tuple $y$ of $T_1$. Let $y$ be
different from $\bf{0}$ in the position $j$. All neighbors of
$\bf{0}$ different from $y$ are neighbors of tuples from $T_1$.
The neighbors that are different from $\bf{0}$ in the $j$th
coordinate position are common neighbors with $y\in T_1$. The
remaining $(q-1)q$ neighbors are exactly covered by $(q-1)q/2$
tuples from $T_1$ at distance 2 from $\bf{0}$. Let $l$ be any free
position for the face $\Gamma_{1 \ldots q-k}^{0 \ldots 0}$, i.e.
any position from $\{q-k+1,\ldots,q+1\}$. Since $k\geq 2$, we
choose $l$ to be distinct from the position $j$.

Consider the tuples that are different from $\bf{0}$ only in the $l$th position. Since $l$ is not $j$ for any $\alpha\in {\cal A}
 \setminus \{0\}$ the  tuple $(0,\ldots,0,$ $\underset{
l}\alpha,0,\ldots,0)$  is a neighbor of a unique tuple from
 $T_1$ at distance 2 from $\bf{0}$. We denote this tuple by $y^{\alpha}$. By Proposition \ref{prop_perf} we have $d_{T_1}=3$.
 Then for any distinct $\alpha$ and $\beta$, $\alpha, \beta \in {\cal
 A}\setminus
 \{0\}$ we see that $$y^{\alpha}=(0,\ldots,0,*,0,0,\ldots,0,\underset{
l}\alpha,0,\ldots\ldots,0)$$ and
$$y^{\beta}=(0,\ldots,0,0,*,0,\ldots,0,\underset{
l}\beta,0,\ldots\ldots,0)$$
 have only one common nonzero coordinate position which is
 $l$. We have $(q-1)$ such tuples $y^\alpha, \alpha\in {\cal
 A}\setminus
 \{0\}$ and at most $q-k$ of them have nonzero positions among nonfree
 positions $\{1,\ldots,q-k\}$ of $\Gamma_{1\ldots q-k}^{0 \ldots 0}$. Then
 there are at least $k-1$ values for $\alpha\in {\cal A}\setminus
 \{0\}$ such that the face
 $\Gamma_{1\ldots q-k\, l}^{0\ldots 0\, \alpha}$ contains a
 tuple $y^{\alpha}\in T_1$. Moreover the face
 $\Gamma_{1\ldots q-k\, l}^{0 \ldots 0 0}$ contains a tuple ${\bf 0}\in
 T_0$. We see that there are $k$ disjoint faces with $k$ free positions that are subsets of
 $\Gamma_{1 \ldots q-k}^{0 \ldots 0}$ and each of them contains at least
 $(k-1)!\cdot 2$ tuples from $T_0\bigcup T_1$ by induction hypothesis. The theorem follows when $k$ is $q$.
\end{proof}

\begin{corollary} The volume of a spherical bitrade in $H(q,q)$ is greater or equal to  $q!/2$.

\end{corollary}
\begin{proof}
The existence of a spherical bitrade in $H(q,q)$ of volume $v$ implies the existence of a perfect bitrade in $H(q+1,q)$ of volume $2v$ by the construction of Proposition \ref{prop_sphperf}. The result follows from the lower bound given by Theorem \ref{thm_perflow}.
\end{proof}

\medskip

{\bf Acknowledgements.} The authors express their gratitude to
Sergey Avgustinovich for critical remarks on the construction from
Theorem \ref{SphericalBitradesAlt}, Vladimir Potapov for providing
a simple argument for the proof of Theorem
\ref{SphericalTradesMDS}, Anna Taranenko, Alexandr Valyuzenich and
Denis Krotov for valuable comments and suggestions.

\medskip


\end{document}